\documentclass{aamas2017}

\usepackage{amsfonts,amsmath}
\usepackage{amssymb}
\usepackage{paralist}
\usepackage{multirow}
\usepackage{xspace}
\usepackage{graphicx}
\usepackage{tikz}
\usepackage[textsize=scriptsize]{todonotes}
\usepackage[noend,ruled]{algorithm2e}
\usepackage{color}
\usepackage{comment}
\usepackage{float}
\usepackage{tikz}
\usepackage{mathtools}
\usetikzlibrary{positioning}
\usepackage{bbm}
\usepackage{caption}
\DeclareCaptionType{copyrightbox}
\usepackage{subcaption}
\usepackage{booktabs}
\usepackage[]{units}

\newtheorem{definition}{Definition}

\newtheorem{theorem}{Theorem}
\newtheorem{proposition}[theorem]{Proposition}

\newcommand{\calR}{\mathcal{R}}
\newcommand{\calS}{\mathcal{S}}
\newcommand{\np}{{\mathrm{NP}}}

\newcommand{\fpt}{{\mathrm{FPT}}}

\newcommand{\xp}{{\mathrm{XP}}}

\newcommand{\p}{{\mathrm{P}}}

\newcommand{\commentPF}[1]{{\color{red}(PF: #1)}}

\title{Bribery as a Measure of Candidate Success: \\ Complexity Results
  for Approval-Based Multiwinner Rules}

\numberofauthors{3}
\author{
Piotr Faliszewski \\
       \affaddr{AGH University}\\
       \affaddr{Krakow, Poland}\\
       \email{faliszew@agh.edu.pl}
\alignauthor
Piotr Skowron \\
       \affaddr{TU Berlin}\\
       \affaddr{Berlin, Germany}\\
       \email{p.k.skowron@gmail.com}
\alignauthor
Nimrod Talmon \\
       \affaddr{Weizmann Institute of Science}\\
       \affaddr{Rehovot, Israel}\\
       \email{nimrodtalmon77@gmail.com}
}

\begin{document}

\maketitle

\begin{abstract}
  We study the problem of bribery in multiwinner elections, for the
  case where the voters cast approval ballots (i.e., sets of
  candidates they approve) and the bribery actions are limited to:
  adding an approval to a vote, deleting an approval from a vote, or
  moving an approval within a vote from one candidate to the other.
  We consider a number of approval-based multiwinner rules (AV, SAV,
  GAV, RAV, approval-based Chamberlin--Courant, and PAV). We find the
  landscape of complexity results quite rich, going from
  polynomial-time algorithms through $\np$-hardness with
  constant-factor approximations, to outright
  inapproximability. Moreover, in general, our problems tend to be
  easier when we limit out bribery actions on increasing the number of
  approvals of the candidate that we want to be in a winning committee
  (i.e., adding approvals only for this preferred candidate, or moving
  approvals only to him or her). We also study parameterized complexity
  of our problems, with a focus on parameterizations by 
   the numbers of  voters or candidates.
%
\end{abstract}


\keywords{Multiwinner elections, Bribery, Approval-Based voting}

\section{Introduction}
No one enjoys losing an election. Nonetheless, it is a natural part of
life and instead of drowning in sorrow, a skillful candidate (or, a
rational agent in a multiagent environment) should rather focus on
learning as much as possible from the defeat. In particular, such a
candidate deserves to know how well he or she did in the election and
how close he or she was to winning. In single-winner elections the
candidates typically receive some scores (e.g., in Plurality
elections, the most popular type of single-winner elections, these
scores are the numbers of voters that consider a given candidate as
the best one) and the highest-scoring candidate is a winner. Then
reporting the scores for all the candidates gives them some idea of
their performance.

This score-reporting approach, however, has a number of
drawbacks. First, for some 
rules either there are no natural notions of the score or ones that
exist do not necessarily give a very good idea of a candidate's level
of success. For example, under the single-winner variant of the STV
rule, the voters rank candidates from the best one to the worst one
and we keep on deleting the candidates with the lowest Plurality score
until there is only one left, the winner. On the surface, the rule
does not assign scores to the candidates.  We could, of course, define
the STV score as the round number in which the candidate is
eliminated, but it would not be very useful: Even a tiny change in the
votes can radically change the elimination order (see, e.g., the work
of Woodall~\cite{woo:j:properties}; the effect is also used in the
hardness proofs of manipulation for
STV~\cite{bar-oli:j:polsci:strategic-voting,wal:c:uncertainty-in-preference-elicitation-aggregation}).

Second, this approach is quite problematic to use within multiwinner
elections, where whole committees of candidates are selected. If there
were $m$ candidates and the committee size were $k$, then one would
have to list $m \choose k$ scores, one for each possible
committee. One possible remedy would be to list for each candidate $p$
the score and the contents of the best committee that included
$p$. Unfortunately, this would not address the first issue, which in
multiwinner elections is even more pressing than in single-winner ones
and, more importantly, would not really tell the candidate what
\emph{this candidate's} performance was, but rather would bind him or
her to some committee.

Third, the scores used by some rules may not be sufficiently
informative. For example, in Copeland elections the score of candidate
$c$ is the number of candidates $d$ for whom a
majority of voters ranks $c$ higher than $d$. Yet, no one would claim
that two candidates with the same Copeland score, where one loses his
or her pairwise contests by just a few votes each and the other loses
them by a huge margin, performed similarly.

Finally, the notion of a score may be quite arbitrary. Going back to
the previous example, the Copeland score can be defined so that a
candidate receives $1$ point for winning a pairwise contest, $-1$
point for losing it, and $0$ points for a tie, but one may as well
define it to give $1$ point for a victory, $0$ points for a loss, and
$0.5$ points for a tie. Both approaches are perfectly appropriate and
both are used in the literature, but when used as measures of a
candidate's success, they need to be interpreted quite differently.


To address the issues mentioned above, we propose an approach based on
the \textsc{Bribery} family of problems, introduced by Faliszewski et
al.~\cite{fal-hem-hem:j:bribery} and then studied by a number of other
authors (see the works of Elkind, Faliszewski, and
Slinko~\cite{elk-fal-sli:c:swap-bribery}, Dorn and
Schlotter~\cite{dor-sch:j:parameterized-swap-bribery}, Bredereck et
al.~\cite{bre-fal-nie-tal:c:multiwinner-shift-bribery}, and
Xia~\cite{xia:margin-of-victory} as some examples; we give a more
detailed discussion in Section~\ref{sec:related}). In these problems
we are allowed to perform some actions that modify the votes and we
ask what is the smallest number of such actions that ensure that a
given candidate is a winner of the election. The fewer actions are
necessary for a particular candidate, the better he or she did in the
election (e.g., the winners require no actions at all).

To present our ideas, we focus on approval-based multiwinner
elections. We are interested in multiwinner elections because for them
measuring the performance of losing candidates is far less obvious
than in the single-winner case, and we focus on approval-based rules
(where each voter provides a set of candidates that he or she approves),
as opposed to rules based on preference orders (where each voter ranks
candidates from best to worst), because \textsc{Bribery}-style
problems for preference-order-based rules are already quite
well-studied~\cite{fal-rot:b:control} 
(even in the multiwinner
setting~\cite{bre-fal-nie-tal:c:multiwinner-shift-bribery}).

Let us now describe our setting more precisely. We assume that we are
given an election, i.e., a set of candidates and a collection of
voters (each voter with a set of candidates that he or she approves
of), and a multiwinner voting rule. Multiwinner rules take as input an
election and a committee size $k$, and output a set of $k$ candidates
that form a winning committee (formally, we assume that they output a
set of tied committees, but for now we disregard this issue). Let us
consider one of the simplest approval-based multiwinner rules, namely
\textsc{Approval Voting} (\textsc{AV} for short; we also consider a
number of other rules throughout the paper).
Under \textsc{AV}, a candidate receives a
point for each voter that approves him or her, and the winning
committee consists of the $k$ highest-scoring candidates. Let us
assume that we have four candidates, $a$, $b$, $c$, and $p$, and nine
voters, $v_1, \ldots, v_9$, who approve the following candidates:
\begin{align*}
  v_1 &= \{a,b,c\}, &  v_2 &= \{b,c\}, &  v_3 &= \{a\}, \\
  v_4 &= \{a,b\},   &  v_5 &= \{a,b\}, & v_6 &= \{a,c\},\\
  v_7 &= \{b,c,p\}, &  v_8 &= \{a\},   &  v_9 &= \{a\}.   
\end{align*}
The scores of $a$, $b$, $c$, and $p$, are, respectively, $7$, $5$,
$4$, and $1$.  (See Figure~\ref{fig:example:av} for a graphical
presentation.)  For size two, the winning committee is $\{a,b\}$.
\begin{figure}
\centering
\begin{subfigure}[b]{0.09\textwidth}

\begin{tikzpicture}
\draw (0.05,0) -- (1.5,0);

\draw (0.1,0.00) rectangle (0.4,0.25);
\draw (0.1,0.25) rectangle (0.4,0.50);
\draw (0.1,0.50) rectangle (0.4,0.75);
\draw (0.1,0.75) rectangle (0.4,1.00);
\draw (0.1,1.00) rectangle (0.4,1.25);
\draw (0.1,1.25) rectangle (0.4,1.50);
\draw (0.1,1.50) rectangle (0.4,1.75);
\node [above] at (0.25,-0.4) {$a$};

\draw (0.45,0.00) rectangle (0.75,0.25);
\draw (0.45,0.25) rectangle (0.75,0.50);
\draw (0.45,0.50) rectangle (0.75,0.75);
\draw (0.45,0.75) rectangle (0.75,1.00);
\draw (0.45,1.00) rectangle (0.75,1.25);
\node [above] at (0.6,-0.4) {$b$};

\draw (0.8,0.00) rectangle (1.10,0.25);
\draw (0.8,0.25) rectangle (1.10,0.50);
\draw (0.8,0.50) rectangle (1.10,0.75);
\draw (0.8,0.75) rectangle (1.10,1.00);
\node [above] at (0.95,-0.4) {$c$};

\draw (1.15,0.00) rectangle (1.45,0.25);
\node [above] at (1.30,-0.45) {$p$};
\end{tikzpicture}

\caption{Original}
\end{subfigure}\;\;
\begin{subfigure}[b]{0.09\textwidth}

\begin{tikzpicture}
\draw (0.05,0) -- (1.5,0);

\draw (0.1,0.00) rectangle (0.4,0.25);
\draw (0.1,0.25) rectangle (0.4,0.50);
\draw (0.1,0.50) rectangle (0.4,0.75);
\draw (0.1,0.75) rectangle (0.4,1.00);
\draw (0.1,1.00) rectangle (0.4,1.25);
\draw (0.1,1.25) rectangle (0.4,1.50);
\draw (0.1,1.50) rectangle (0.4,1.75);
\node [above] at (0.25,-0.4) {$a$};

\draw (0.45,0.00) rectangle (0.75,0.25);
\draw (0.45,0.25) rectangle (0.75,0.50);
\draw (0.45,0.50) rectangle (0.75,0.75);
\draw (0.45,0.75) rectangle (0.75,1.00);
\draw (0.45,1.00) rectangle (0.75,1.25);
\node [above] at (0.6,-0.4) {$b$};

\draw (0.8,0.00) rectangle (1.10,0.25);
\draw (0.8,0.25) rectangle (1.10,0.50);
\draw (0.8,0.50) rectangle (1.10,0.75);
\draw (0.8,0.75) rectangle (1.10,1.00);
\node [above] at (0.95,-0.4) {$c$};

\draw (1.15,0.00) rectangle (1.45,0.25);
\draw [very thick] (1.15,0.25) rectangle (1.45,0.50);
\draw [very thick] (1.15,0.50) rectangle (1.45,0.75);
\draw [very thick] (1.15,0.75) rectangle (1.45,1.00);
\draw [very thick] (1.15,1.00) rectangle (1.45,1.25);

\node [above] at (1.30,-0.45) {$p$};
\end{tikzpicture}

\caption{Adding}
\end{subfigure}\;\;
\begin{subfigure}[b]{0.09\textwidth}

\begin{tikzpicture}
\draw (0.05,0) -- (1.5,0);

\draw (0.1,0.00) rectangle (0.4,0.25);
\draw (0.1,0.25) rectangle (0.4,0.50);
\draw (0.1,0.50) rectangle (0.4,0.75);
\draw (0.1,0.75) rectangle (0.4,1.00);
\draw (0.1,1.00) rectangle (0.4,1.25);
\draw (0.1,1.25) rectangle (0.4,1.50);
\draw (0.1,1.50) rectangle (0.4,1.75);
\node [above] at (0.25,-0.4) {$a$};

\draw (0.45,0.00) rectangle (0.75,0.25);
\draw [dotted](0.45,0.25) rectangle (0.75,0.50);
\draw [dotted](0.45,0.50) rectangle (0.75,0.75);
\draw [dotted](0.45,0.75) rectangle (0.75,1.00);
\draw [dotted](0.45,1.00) rectangle (0.75,1.25);
\node [above] at (0.6,-0.4) {$b$};

\draw (0.8,0.00) rectangle (1.10,0.25);
\draw [dotted] (0.8,0.25) rectangle (1.10,0.50);
\draw [dotted](0.8,0.50) rectangle (1.10,0.75);
\draw [dotted](0.8,0.75) rectangle (1.10,1.00);
\node [above] at (0.95,-0.4) {$c$};

\draw (1.15,0.00) rectangle (1.45,0.25);

\node [above] at (1.30,-0.45) {$p$};
\end{tikzpicture}

\caption{Deleting}
\end{subfigure}\;\;
\begin{subfigure}[b]{0.099\textwidth}

\begin{tikzpicture}
\draw (0.05,0) -- (1.5,0);

\draw (0.1,0.00) rectangle (0.4,0.25);
\draw (0.1,0.25) rectangle (0.4,0.50);
\draw (0.1,0.50) rectangle (0.4,0.75);
\draw (0.1,0.75) rectangle (0.4,1.00);
\draw (0.1,1.00) rectangle (0.4,1.25);
\draw (0.1,1.25) rectangle (0.4,1.50);
\draw (0.1,1.50) rectangle (0.4,1.75);
\node [above] at (0.25,-0.4) {$a$};

\draw (0.45,0.00) rectangle (0.75,0.25);
\draw (0.45,0.25) rectangle (0.75,0.50);
\draw (0.45,0.50) rectangle (0.75,0.75);
\draw [dotted](0.45,0.75) rectangle (0.75,1.00);
\draw [dotted](0.45,1.00) rectangle (0.75,1.25);
\node [above] at (0.6,-0.4) {$b$};

\draw (0.8,0.00) rectangle (1.10,0.25);
\draw (0.8,0.25) rectangle (1.10,0.50);
\draw (0.8,0.50) rectangle (1.10,0.75);
\draw [dotted](0.8,0.75) rectangle (1.10,1.00);
\node [above] at (0.95,-0.4) {$c$};

\draw (1.15,0.00) rectangle (1.450,0.25);
\draw [very thick] (1.15,0.25) rectangle (1.45,0.50);
\draw [very thick] (1.15,0.50) rectangle (1.45,0.75);
\draw [very thick] (1.15,0.75) rectangle (1.45,1.00);

\node [above] at (1.30,-0.45) {$p$};
\end{tikzpicture}

\caption{Swapping}
\end{subfigure}

\caption{\label{fig:example:av}Approval scores in the election from
  the introduction, together with an illustration show how $d$ can join
  the size-2 committee by, respectively, adding, deleting, or
  swapping approvals.}
\end{figure}
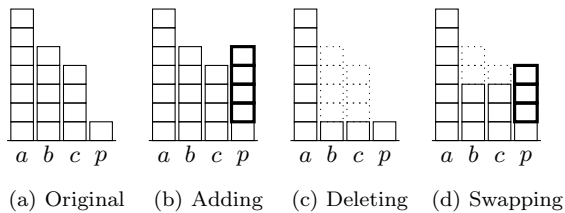
We analyze the performance of $p$ by considering the following three
types of bribery actions:
\begin{description}
\item[Adding Approvals.] In this case, we are allowed to add
  candidates to the voters' sets of approved candidates, paying a unit
  price for each addition.  In our example it suffices to add four
  approvals for $p$ (let us assume that we break ties in favor of
  $p$). Indeed, for the case of \textsc{AV}, this number is the
  difference between the scores of $p$ and the lowest scoring
  committee member.

\item[Deleting Approvals.] In this case we are allowed to remove
  approvals.  In our example one has to remove seven approvals. While
  deleting approvals may not seem an intuitively good measure of
  candidate's performance, in fact it behaves quite interestingly. As
  opposed to 
  adding approvals, not only does it measure how many points a
  candidate is missing to join the committee, but also
  it accounts for the number of committee non-members that did better
  than $p$.

\item[Swapping Approvals.] Here we are allowed to move approvals
  between candidates within each vote. In our example it suffices to
  move three approvals (two from $b$ to $p$, e.g., in the votes $v_1$
  and $v_2$, and one from $c$ to $p$, e.g., in $v_6$). This measure
  seems to be somewhere between the previous two. It takes into
  account the score difference between the lowest-scoring committee
  member and $p$, the number of candidates with scores in between, and
  how the approvals are distributed between the votes (within a single
  vote, we can swap only one approval to $p$).
\end{description}

Indeed,
the above interpretations are particularly easy and natural for
\textsc{AV}, but the numbers of approvals that one has to add, delete,
or swap to ensure a particular candidate's victory are useful measures
for other rules as well.  

Unfortunately, many bribery problems are known to be $\np$-hard. In
this paper we study the complexity of our three variants of bribery
under approval-based rules (\textsc{AddApprovals-Bribery},
\textsc{DeleteApprovals-Bribery}, and \textsc{SwapApprovals-Bribery})
in the following settings: Either each bribery action comes at unit
price (as in the examples above) or each bribery action has a separate
price (this can be used, e.g., to model certain knowledge about some
voters, such as the fact that some voters would never approve our
candidate, or would never delete any approvals), and either we allow
all possible actions, or only those that increase the number of
approvals of our preferred candidate (this restriction does not apply
to the case of deleting approvals). We obtain the following results:
\begin{enumerate}
\item Most of our problems turn out to be $\np$-hard for most of our
  rules (the exceptions include AV in most settings, and GAV and RAV
  when adding unit-cost approvals for the preferred candidate).


\item Problems where bribery actions are focused on the preferred
  candidate tend to be easier than the unrestricted ones (e.g., we
  sometimes obtain $2$-approximation algorithms instead of
  inapproximability results, or $\fpt$ algorithms instead of $\xp$
  ones). Focusing on unit prices has a similar effect.




\item Most of our problems are in $\fpt$
  parameterized either by the number of candidates or the number of voters.
  
\end{enumerate}

Due to space restrictions, we omit many of the proofs (available upon
request). We included proofs that we felt were most illustrative of
the techniques used, and were not based on those already in the
literature. We discuss related literature in
Section~\ref{sec:related}.


\section{Preliminaries}\label{sec:prelim}

An approval-based election $(C,V)$ consists of a set $C$ of candidates
and a collection $V = (v_1, \ldots, v_n)$ of voters.  Each voter has a
set of candidates that he or she approves, and---by a slight abuse of
notation---we refer to these sets through the voters' names (e.g., we
write $v_1 = \{a,b\}$ to indicate that voter $v_1$ approves candidates
$a$ and $b$).

A multiwinner voting rule is a function $\calR$ that given an election
$E = (C,V)$ and a number $k$ ($1 \leq k \leq |C|$) returns a nonempty
family of committees (i.e., size-$k$ subsets of $C$).  We treat each
committee in $\calR(E,k)$ as tied for winning.  (Tie-breaking can
sometimes affect the complexity of voting
problems~\cite{con-rog-xia:c:mle,obr-elk-haz:c:ties-matter,obr-elk:c:random-ties-matter};
our approach is frequently taken in the literature as a simplifying
assumption~\cite{fal-rot:b:control}).\medskip


\noindent\textbf{Approval-Based Rules.}\quad
Let $(C,V)$ be an election and let $k$ be the desired committee
size. Following Aziz et
al.~\cite{azi-gas-gud-mac-mat-wal:c:multiwinner-approval,azi-bri-con-elk-fre-wal:j:justified-representation},
we consider the following six rules (unless we mention otherwise, for
each rule described below there is a simple, natural polynomial-time
winner determination algorithm):


\begin{description}
\item[Approval Voting (AV).] Under the AV rule, the score of each
  candidate is the number of voters that approve him or her. Winning
  committees are those that contain $k$ candidates with highest
  scores. 

\item[Satisfaction Approval Voting (SAV).] Under the SAV rule, each
  voter $v_i$ gives $\frac{1}{|v_i|}$ points to each of his or her
  approved candidates (i.e., each voter is given one point that he or
  she distributes equally among the approved candidates). Winning
  committees consist of $k$ candidates with highest total scores.

\item[Chamberlin--Courant Approval Voting (CCAV).]  We say that a
  voter approves a given committee if he or she approves at least one
  member of this committee. CCAV selects those committees that are
  approved by the largest number of voters.  (One interpretation is
  that voters get representatives from the committee; a voter who
  approves the committee can be represented well).  Unfortunately,
  computing a winning committee under CCAV is
  $\np$-hard~\cite{pro-ros-zoh:j:proportional-representation,bet-sli-uhl:j:mon-cc}.

\item[Greedy Approval Voting (GAV).] GAV was considered by Lu and
  Boutlier~\cite{bou-lu:c:chamberlin-courant} as an approximation
  algorithm for CCAV, but it turned out to be an interesting rule on
  its
  own~\cite{elk-fal-sko-sli:j:multiwinner-properties,azi-bri-con-elk-fre-wal:j:justified-representation}.\footnote{Strictly
    speaking, Lu and Boutilier~\cite{bou-lu:c:chamberlin-courant} and
    Elkind et al.~\cite{elk-fal-sko-sli:j:multiwinner-properties}
    discuss a variant of the algorithm that uses preference orders and
    Borda scores. Nonetheless, their main conclusions transfer to the
    approval-based case.} GAV starts with an empty committee $W$ and
  executes $k$ rounds, where in each round it adds to $W$ a candidate
  that maximizes the number of voters who approve at least one member
  of $W$ (in case of a tie, we assume that there is a fixed
  tie-breaking rule; thus GAV always returns a single committee). If a
  winning committee under CCAV is approved by $\mathrm{OPT}$ voters,
  then the committee produced by GAV is approved by at least
  $(1-\nicefrac{1}{e})\mathrm{OPT}$
  voters~\cite{bou-lu:c:chamberlin-courant}.

\item[Proportional Approval Voting (PAV).] Under PAV, voter $v_i$ assigns to committee $W$ score $\sum_{t=1}^{|W \cap
    v_i|}\frac{1}{t}$. PAV outputs those committees that receive the
  highest total score from the voters.  The rule is $\np$-hard to
  compute~\cite{sko-fal-lan:j:owa,azi-gas-gud-mac-mat-wal:c:multiwinner-approval},
  but---as shown by Aziz et
  al.~\cite{azi-bri-con-elk-fre-wal:j:justified-representation}---satisfies
  strong axiomatic properties, making it well-suited for choosing
  representative bodies (e.g., parliaments, university senates, etc.;
  GAV and CCAV satisfy weaker variants of these properties).

\item[Reweighted Approval Voting (RAV).] RAV relates to PAV in the
  same way as GAV relates to CCAV. It starts with an empty committee
  $W$ and proceeds in $k$ rounds, in each adding to the committee a
  candidate that maximizes the PAV score of the committee. It
  guarantees finding a committee whose score is at least a
  $(1-\nicefrac{1}{e})$ fraction of that of a PAV winning committee.
  
\end{description}

Naturally, the above rules have different strengths and weaknesses,
and should be applied in different
settings~\cite{azi-bri-con-elk-fre-wal:j:justified-representation,elk-fal-sko-sli:j:multiwinner-properties,sko-fal-lan:j:owa}.
We provide more pointers regarding their properties and history in
Section~\ref{sec:related}.

\begin{table*}[t]

\newcommand{\twocol}[1]{\multicolumn{2}{c|}{#1}}
\newcommand{\twocoln}[1]{\multicolumn{2}{c}{#1}}
\newcommand{\tabp}{\multicolumn{2}{c|}{\multirow{4}{*}{$\p$}}}
\small
\centering
\begin{tabular}{c|cc|cc|cc|cc|cc}
\toprule[0.75pt]
   \scriptsize\bf operation
   & \multicolumn{4}{c|}{\bf Adding Approvals} &  \multicolumn{2}{c|}{\bf Deleting Approvals} & \multicolumn{4}{c}{\bf Swapping Approvals} \\
   \scriptsize\bf restriction
   & \multicolumn{2}{c|}{\scriptsize\bf (for $\boldsymbol{p}$)} & \multicolumn{2}{c|}{\scriptsize\bf (none)}  
   & \multicolumn{1}{c}{\scriptsize\bf } & \multicolumn{1}{c|}{\scriptsize\bf }  
   & \multicolumn{2}{c|}{\scriptsize\bf (for $\boldsymbol{p}$)} & \multicolumn{2}{c}{\scriptsize\bf (none)}  \\

   \scriptsize\bf {prices}
   & {\scriptsize\bf (unit)} & \scriptsize\bf (any) 
   & {\scriptsize\bf (unit)} & \scriptsize\bf (any) 
   & {\scriptsize\bf (unit)} & \scriptsize\bf (any) 
   & {\scriptsize\bf (unit)} & \scriptsize\bf (any) 
   & {\scriptsize\bf (unit)} & \scriptsize\bf (any) \\
   \midrule[0.75pt]
   \multirow{4}{*}{AV}
       & \tabp& \tabp & \tabp & \multirow{4}{*}{$\p$} & $\np$-hard & \multirow{4}{*}{$\p$} & $\np$-hard  \\
       &   &  &   &   &   &   &                       & inapprox.&                       & inapprox. \\
       &   &  &   &   &   &   &                       & $\fpt(m)$  &                       & $\fpt(m)$ \\
       &   &  &   &   &   &   &                       & $\fpt(n)$  &                       & $\xp(n)$ \\

   \midrule
   \multirow{4}{*}{SAV}  & \twocol{$\np$-hard}  & \twocol{$\np$-hard}      & \twocol{$\np$-hard} & \twocol{$\np$-hard} & \twocoln{$\np$-hard} \\
                         & \twocol{$2$-approx.} & $2$-approx. & inapprox.  & \twocol{inapprox.}  & \twocol{inapprox.}  & ? & inapprox.  \\
                         & \twocol{$\fpt(m)$}   & \twocol{$\fpt(m)$}       & \twocol{$\fpt(m)$}  & $\fpt(m)$   & $\xp(m)$       & $\fpt(m)$   & $\xp(m)$   \\ 
                         & \twocol{$\fpt(n)$}   & $\fpt(n)$ & $\xp(n)$        & \twocol{$\xp(n)$}      & \twocol{$\fpt(n)$}        & $\fpt(n)$ & $\xp(n)$             \\
   \midrule
   \multirow{4}{*}{GAV}  & \tabp      & \twocol{$\np$-hard}         & \twocol{$\np$-hard}  & \twocol{$\np$-hard} & \twocoln{$\np$-hard} \\
                         &                      &                   & \twocol{inapprox.}  & \twocol{inapprox.}   & \twocol{inapprox.}  & \twocoln{inapprox.}  \\
                         &                      &                   & \twocol{$\fpt(m)$}  & \twocol{$\fpt(m)$}   & $\fpt(m)$      & $\xp(m)$  & $\fpt(m)$ & $\xp(m)$  \\ 
                         &                      &                   & \twocol{$\fpt(n)$}  & \twocol{$\fpt(n)$}   & \twocol{$\fpt(n)$}  & $\fpt(n)$ & $\xp(n)$        \\
   \midrule
   \multirow{4}{*}{RAV}  & \multirow{4}{*}{$\p$}& ?                 & \twocol{$\np$-hard} & \twocol{$\np$-hard}  & \twocol{$\np$-hard} & \twocoln{$\np$-hard} \\
                         &                      & PTAS              & \twocol{inapprox.}  & \twocol{inapprox.}   & \twocol{inapprox.}  & \twocoln{inapprox.}  \\
                         &                      & {$\fpt(m)$}       & \twocol{$\fpt(m)$}  & \twocol{$\fpt(m)$}   & $\fpt(m)$     & $\xp(m)$   & $\fpt(m)$ & $\xp(m)$  \\ 
                         &                      & {$\fpt(n)$}       & $\fpt(n)$  & $\xp(n)$      & \twocol{$\xp(n)$}       & \twocol{$\fpt(n)$}      & $\fpt(n)$ & $\xp(n)$          \\


   \midrule[0.75pt]
   \multirow{2}{*}{CCAV}  & \twocol{$\fpt(m)$}& \twocol{$\fpt(m)$}& \twocol{$\fpt(m)$} & $\fpt(m)$   & $\xp(m)$ & $\fpt(m)$ & $\xp(m)$\\
                           & \twocol{$\fpt(n)$}& \twocol{$\fpt(n)$}& \twocol{$\fpt(n)$}& \twocol{$\fpt(n)$} & $\fpt(n)$ & $\xp(n)$ \\
   \midrule
   \multirow{2}{*}{PAV}    & \twocol{$\fpt(m)$} & \twocol{$\fpt(m)$} & \twocol{$\fpt(m)$}  & $\fpt(m)$   & $\xp(m)$  & $\fpt(m)$ & $\xp(m)$\\
                           & \twocol{$\fpt(n)$}& $\fpt(n)$ & $\xp(n)$     & \twocol{$\xp(n)$}     & \twocol{$\fpt(n)$}     & $\fpt(n)$ & $\xp(n)$     \\

\bottomrule[0.75pt]

\end{tabular}
\caption{\label{tab:results}Results for all our rules and all variants of the problems. 
For each rule and each scenario we report four entries: 
  (1) is the problem in $\p$ or is it $\np$-hard, 
  (2) what is the best known approximation algorithm, 
  (3 and 4) what are the best known parameterized algorithm for parameterization by 
     the number of candidates ($m$) and the number of voters ($n$), respectively.
  Each cell in the table is divided into two columns, one for the unpriced variant of the 
  problem and one for the the priced variant. When a  result for both columns is the same, 
  we write it in the middle of the cell.}

\end{table*}


\medskip
\noindent\textbf{Bribery Problems}.\quad
We are interested in bribery problems where we can perform the
following types of operations:
\begin{itemize}
\item[\textsc{AddApprovals}:] A single operation means adding an
  approval for a given candidate in a given vote.
\item[\textsc{DeleteApprovals}:] A single operation means removing an
  approval from a given candidate in a given vote.
\item[\textsc{SwapApprovals}:] A single operation means moving an
  approval from a given candidate in a given vote to another
  candidate---originally not approved---within the same vote.
\end{itemize}

In the basic variant of our problems, each operation comes with the
same, unit price.
\begin{definition}\label{def:problem}
  Let $\calR$ be an approval-based multiwinner rule and let
  \textsc{Op} be one of \textsc{AddApprovals},
  \textsc{DeleteApprovals}, or \textsc{SwapApprovals}. In the
  $\calR$-\textsc{Op-Bribery} problem, we are given an election
  $(C,V)$, a preferred candidate $p \in C$, and two integers, the
  committee size $k$ and the budget $b$. We ask whether it is possible
  to ensure that $p$ belongs to at least one $\calR$-winning committee
  of size $k$ by applying at most $b$ operations of type \textsc{Op}
  to election $(C, V)$.
\end{definition}

We follow Bredereck et
al.~\cite{bre-fal-nie-tal:c:multiwinner-shift-bribery} in that it
suffices for $p$ to belong to just one of the winning committees.
(The approach where $p$ should belong to every winning committee---as
in the work of Meir et
al.~\cite{mei-pro-ros-zoh:j:multiwinner}---would be as natural.)

While Definition~\ref{def:problem} gives the baseline variants of our
problems, we also consider two modifications. In the priced variant
(denoted by operations \textsc{\$AddApprovals},
\textsc{\$DeleteApprovals}, and \textsc{\$SwapApprovals}), we assume
that each possible operation comes with a distinct price (that depends
both on the voter and on the candidate(s) to which it applies; e.g.,
adding an approval for $p$ to some vote $v$ could cost $10$ units,
whereas adding an approval for some other candidate $c$ to the same vote
$v$ could cost $2$ units). That is, our problems are closer to
\textsc{Swap Bribery} and \textsc{Shift Bribery} of Elkind et
al.~\cite{elk-fal-sli:c:swap-bribery,elk-fal:c:shift-bribery} than to
\textsc{Bribery} and \textsc{\$Bribery} of Faliszewski et
al.~\cite{fal-hem-hem:j:bribery} (where upon paying a voter's price,
one can modify the vote arbitrarily).

We also distinguish variants of the \textsc{(\$)AddApprovals} and
\textsc{(\$)SwapApprovals} operations where one is limited to,
respectively, adding approvals only for $p$ or swapping approvals only
to $p$. These problems model natural, positive scenarios, where we
want to find out what support candidate $p$ should have garnered to
win the election.

\section{Results}

For all our rules and settings, we seek results of three kinds.
First, we check whether the problem is polynomial-time solvable (few
rare cases) or is $\np$-hard (typical). Then, to deal with
$\np$-hardness, we seek approximation and $\fpt$ algorithms.
Unfortunately, in most cases we show that our problems are hard to
approximate in polynomial time within any constant factor.  For the
case of parameterized complexity, we show that all of our problems are
fixed-parameter tractable (in $\fpt$), provided that we consider unit
prices and take as the parameter either the number of candidates or
the number of voters. For the case of priced elections, we still get a
fairly comprehensive set of $\fpt$ algorithms, but we do miss some
cases (and we resort to $\xp$ algorithms then; nonetheless, we
strongly believe that new proof techniques would lead to $\fpt$
results for all our cases).

We summarize our results in Table~\ref{tab:results}.  Below we first
study our polynomial-time computable rules (AV, SAV, GAV, and RAV),
for which we prove $\p$-membership, $\np$-hardness, and
(in)approximability results, and then move on to parameterized
complexity, where we consider all the rules.

\subsection{The Easy Case: Approval Voting}

For AV, almost all our problems can be solved in polynomial time using
simple greedy algorithms.

\begin{theorem}
  Let \textsc{Op} be one of \textsc{(\$)AddApprovals},
  \textsc{(\$)DeleteApprovals}, and
  \textsc{SwapApprovals}. \textsc{AV-Op-Bribery} is in $\p$ (also for
  the cases where we can add/swap approvals only to $p$).
\end{theorem}
\begin{proof}[sketch] 
  Let $(C,V)$ be the input election, $p$ be the preferred candidate,
  $k$ be the committee size, and $b$ be the budget.  For the case of
  (priced) bribery by adding approvals, it suffices to keep on adding
  approvals for $p$ 
  in the order of nondecreasing price of this operation, until either
  $p$ becomes a member of some winning committee or we exceed the
  budget (adding approvals for others is never beneficial). 

  For the case of (priced) bribery by deleting approvals, if $p$ is
  not a winner already then we proceed as follows. Let $C'$ be the set
  of candidates that have more approvals than~$p$. By ``bringing a
  candidate $c \in C'$ down'' we mean the cheapest sequence of
  approval-deletions that ensures that $c$ has the same number of
  approvals as $p$ has (we refer to the total cost of this sequence as
  the cost of bringing $c$ down). We keep on bringing candidates from
  $C'$ down (in the order of nondecreasing cost of this operation)
  until 
  $p$ becomes a member of some winning committee or we exceed the
  budget.

  For the case of (unpriced) bribery by swapping approvals, we first
  guess a threshold $T$ ($0 \leq T \leq |V|$) and then repeat the
  following steps until either $p$ belongs to some winning committee
  or we exceed the budget (if we exceed the budget for every choice of
  $T$, then we reject): We let $C'$ be the set of candidates who have
  more approvals than $p$, except the $k-1$ candidates approved by
  most voters (with ties broken arbitrarily, but in the same way in
  each iteration; this works since we consider unit prices).  We
  remove from $C'$ those candidates who are approved by at most $T$
  voters. Then, if $C'$ is nonempty, we move an approval from some $c
  \in C'$ to $p$ (there is a vote where it is possible because $c$ has
  more approvals than $p$). If $C'$ is empty, then we move an approval
  to $p$ from some arbitrarily chosen candidate in some arbitrarily
  chosen vote.  Intuitively, in this algorithm we guess the score $T$
  that we promise $p$ will have upon entering the winning committee,
  and we keep on moving approvals from ``the most fragile'' opponents
  to $p$, so their scores drop to $T$, whereas $p$'s score increases
  to $T$.
%
\end{proof}

Unfortunately, AV-\textsc{\$SwapApprovals-Bribery} is $\np$-hard and
hard to approximate within any constant factor. This hardness comes
from the fact that when swaps have prices, then it does not suffice to
simply know that there will be \emph{some swap} to perform (as in the
algorithm above) and one cheap swap may prevent another, more useful,
one.


\begin{theorem}\label{theorem:av_dollar_swap}
  \textsc{AV-$\$$SwapApprovals-Bribery} is $\np$-hard, even if we are
  allowed to swap approvals to the preferred candidate only.
\end{theorem}

\begin{proof}
  We reduce from the \textsc{Independent Set} problem, where we are
  given a graph $G$ and an integer $h$, and we ask if there is a set
  of $h$ pairwise non-adjacent vertices in $G$. \textsc{Independent
    Set} is known to be NP-hard even on cubic graphs, i.e., graphs
  with vertices of degree three~\cite{GJ79}.

  Let $(G,h)$ be an instance of \textsc{Independent Set}, where $G$ is
  a cubic graph with $n$ vertices. We construct an instance for
  \textsc{AV-$\$$SwapApprovals-Bribery}, as follows.  We let the
  candidate set be $C = \{p\} \cup \{c_v \mid v$ is a vertex of $G\}$,
  where $p$ is the preferred candidate.  For each edge $e = \{u, v\}$
  in $G$, we introduce a voter $v_e$ who approves the candidates $c_u$
  and $c_v$. For each of these edge voters, each approval swap has
  unit cost. We introduce further $3h$ voters, each approving all the
  vertex candidates; all the swaps for these voters cost $3h+1$.
  Finally, we set the committee size to $k = n-h+1$ and the
  budget to $b = 3h$. This completes the construction which
  can be computed in polynomial time.

  Prior to any approval swaps, $p$ has score zero and every other
  candidate has score $3+3h$ (each vertex touches three edges, and we
  get $3h$ points from the second group of voters).

  If there is a set $\mathit{IS}$ of $h$ pairwise non-adjacent
  vertices of $G$, then we can ensure that $p$ belongs to some winning
  committee: It suffices that for each vertex $v \in \mathit{IS}$, we
  move the approval from $c_v$ to $p$ for the three edge voters that
  correspond to the edges touching $v$ (this is possible because
  $\mathit{IS}$ is an independent set). As a result, $p$'s score
  increases to $3h$, the scores of the $h$ candidates corresponding to
  the vertices from $\mathit{IS}$ drop to $3h$, and so $C \setminus \{
  c_v \mid v \in \mathit{IS}\}$ is a winning committee (and contains
  $p$).


  For the other direction, note that
  (1) the score needed for $p$ is $3h$,
  (2) this score is achieved only if we swap for $p$ in each swap operation,
  and
  (3) if $p$ is to be a member of some winning committee then
      at least $h$ candidates have to lose at least three approvals each.
  It follows that these $h$ candidates have to form an
  independent set because otherwise we would not be able to perform
  all the approval swaps.
\end{proof}

Inapproximability results follow by similar proofs.

\begin{theorem}
  For each $\alpha > 1$, if $\p \neq \np$ then there is no
  polynomial-time $\alpha$-approximation algorithm for
  \textsc{AV-$\$$SwapApprovals-Bribery} (even if we focus on swapping
  approvals to $p$ only).
\end{theorem}

\subsection{Chance for Approximation: SAV}

On the surface, SAV is very similar to AV. Yet, the fact that adding
or deleting a single approval can affect many candidates at the same
time (by decreasing or increasing their share of a voter's point) can
be leveraged to show $\np$-hardness of all our problems.

\begin{theorem}
  Let \textsc{Op} be one of \textsc{(\$)AddApprovals},
  \textsc{(\$)DeleteApprovals}, and
  \textsc{(\$)SwapApprovals}. \textsc{SAV-Op-Bribery} is $\np$-hard
  (also for the cases where we can add/swap approvals only to $p$).
\end{theorem}

Fortunately, not all is lost. Using the general technique of Elkind et
al.~\cite{elk-fal:c:shift-bribery,elk-fal-sli:c:swap-bribery}, we
obtain a $2$-approximation algorithm for the (priced) variant of
adding approvals for $p$ only. To employ the approach of Elkind et
al.~\cite{elk-fal:c:shift-bribery,elk-fal-sli:c:swap-bribery}, it must
be the case that (1) after each bribery action, each non-preferred
candidate $c$ loses at most as many points as the preferred one gains,
(2) there is a pseudo-polynomial time algorithm that computes a
bribery action maximizing the score of $p$ for a given budget, and (c)
if $X$ and $Y$ are two sets of legal bribery actions (i.e., all
bribery actions from $X$ can be executed jointly, and all actions from
$Y$ can be executed jointly,) then $X \cup Y$ also is a legal set of
bribery actions. These conditions hold for
SAV-\textsc{(\$)AddApprovals-Bribery} (for adding approvals to $p$
only) and we get the following result.

\begin{theorem}
  There is a $2$-approximation polynomial-time  algorithm for
  \textsc{SAV-(\$)AddApprovals-Bribery} for the case where we add
  approvals to $p$ only.
\end{theorem}

The theorem also works for SAV-\textsc{AddApprovals-Bribery} (i.e.,
for the unrestricted, unpriced case) because if there is a solution
that adds approvals to some candidates other than $p$, then there is
also one with the same cost or lower that adds approvals to $p$
only. (If we add an approval for some candidate $c$, $c \neq p$, in a
vote where $p$ is not approved, then it is better to add the approval
to $p$. If we add an approval in a vote where $p$ already is approved,
then it is better to not make this addition.)

On the other hand, the above technique does not apply to
SAV-\textsc{\$AddApprovals-Bribery} (e.g., there are bribery actions
that do not increase the score of the preferred candidate but decrease
the scores of others, which breaks condition (1) above) and, indeed,
we obtain inapproximability.

\begin{theorem}\label{theorem:sav_add_inapprox}
  For each $\alpha > 1$, if $\p \neq \np$ then there is no
  polynomial-time $\alpha$-approximation algorithm for
  \textsc{SAV-$\$$AddApprovals-Bribery}.
\end{theorem}

The proof follows by noting that the classic \textsc{SetCover} problem
(which is not approximable within any constant factor when $\p \neq
\np$) can be embedded within \textsc{SAV-$\$$AddApprovals-Bribery}.
The key idea is to model each set from a \textsc{SetCover} instance as
a voter.  Due to the nature of SAV, as soon as we add an approval to a
vote, the scores of all the previously approved candidates (who
correspond to elements) decrease.  Our construction guarantees that to
make $p$ winner, one needs to decrease the score of all
element-candidates and, 
thus, adding an approval to a ``set voter'' can be viewed as covering
the elements from the corresponding set.  It is possible to provide
such construction which preserves the inapproximability bound of
\textsc{SetCover}.

The proof for 
\textsc{SAV-(\$)DeleteApprovals-Bribery}
relies on similar tricks, but is far more involved (again, we cannot
use the $2$-approximation technique because deleting an approval for a
candidate decreases his or her score more than it increases the score
of the preferred candidate).

\begin{theorem}
  For each $\alpha > 1$, if $\p \neq \np$ then there is no
  polynomial-time $\alpha$-approximation algorithm for
  \textsc{SAV-$\$$DeleteApprovals-Bribery}.
\end{theorem}

The case of swapping approvals is more tricky. We cannot use the
$2$-approximation trick, because 
if $X$ and $Y$ are two sets of approval-swaps to perform (each
possible to execute) then $X \cup Y$ may be impossible to perform
(e.g., it may require to move an approval to the preferred candidate
within some vote from two different candidates).  In fact, for the
case where we only move approvals to the preferred candidate, we
obtain outright inapproximability result (which immediately translates
to the unrestricted, priced setting; with high prices we can enforce
approval-swaps to $p$ only). The general result for unit-price swaps
remains elusive.

\begin{theorem}
  For each $\alpha > 1$, if $\p \neq \np$ then there is no
  polynomial-time $\alpha$-approximation algorithm for
  \textsc{SAV-(\$)SwapApprovals-Bribery} for the case where we only
  move approvals to $p$.
\end{theorem}
\begin{proof}
  Let us fix $\alpha$ to be a positive integer, $\alpha \geq 1$.  We
  will give a reduction $f$ from a restricted variant of the
  \textsc{X3C} problem to \textsc{SAV-SwapApprovals-Bribery} (for the
  case where we can move approvals to $p$ only) and argue that an
  $\alpha$-approximation algorithm for the latter would have to decide
  the former.
  In our \textsc{Restricted X3C} we are given a set $X = \{x_1,
  \ldots, x_{3n}\}$ of elements and a family $\calS = \{S_1, \ldots,
  S_{3n}\}$ of sets, such that (a) each set contains exactly three
  elements, and (b) each element belongs to exactly three sets. 
  We ask if there is a family of $n$ sets from $\calS$ whose union is
  exactly $X$. This variant 
  remains
  $\np$-hard~\cite{gonzalez1985clustering}.

  Let $I$ be an instance of \textsc{Restricted X3C} (with input as
  described above).  We set $N = 27(\alpha n + 1)$ (intuitively, $N$
  is simply a value much larger than $n$) and we form an instance of
  our problem as follows. We let the candidate set be $C = \calS \cup
  D \cup \{p\}$, where $D = \{d_1, \ldots, d_N\}$ is a set of dummy
  candidates needed for our construction, and we introduce the
  following voters:
  \begin{enumerate}
  \item For each element $x_i \in X$, we introduce one voter $v_i$
    that approves the three set-candidates $S_{j'}$, $S_{j''}$,
    $S_{j'''}$ that correspond to the sets that contain $x_i$. We
    refer to these voters as element voters.

  \item We introduce $n \cdot (N+3n) - 1$ voters, each approving all
    the candidates from $\calS$ and $D$. We write $V'$ to denote the
    set of these voters.

  \item We introduce $10nN$ voters, each approving all the candidates
    in $D$. We denote the set of these voters by $V''$.

  \end{enumerate}
  Prior to bribery, $p$ has score $0$, each set candidate has score
  $1+n - \frac{1}{N+3n}$, and each dummy candidate has score at least
  $10n$.  We set the committee size to $k = N + 2n + 1$, and the
  budget to $b = 3n$.

  If $I$ is a ``yes''-instance, then it is possible to ensure that $p$
  belongs to some winning committee using at most $3n$ approval swaps:
  For each set $S_j$ from the exact cover we take all voters
  corresponding to elements covered by $S_j$ and for these voters we
  move approvals from $S_j$ to $p$. Consequently, $p$ is approved by
  all the voters corresponding to elements of $X$ and obtains
  $\nicefrac{1}{3} \cdot 3n = n$ points. Since the score of each of
  the sets from the exact cover drops to $n - \nicefrac{1}{N+3n} < n$,
  there are $n$ candidates with score lower than $p$.  In effect, $p$
  belongs to a winning committee.

  Now, consider what happens if $I$ is a ``no''-instance. After
  $3\alpha n$ swaps, the score of $p$ can be at most
  $n + \nicefrac{(3\alpha n-3n)}{N} 
  < n + \nicefrac{1}{9}$ (at best, we can get $n$ points from the
  element voters using $3n$ swaps, and use the remaining $3\alpha
  n-3n$ swaps for voters in $V''$).
  Since there is no exact cover, after executing all the swaps there
  are at most $n-1$ set candidates such that no element voter
  approves them.  Every other set candidate is approved by at least
  one element voter and at least $n \cdot (N+3n) -1 - 3\alpha n$
  voters from $V'$. The score of such candidate is, thus, at least:
  \begin{align*}
    & \nicefrac{1}{3} + \big(n \cdot (N+3n) -1 - 3\alpha n \big) \cdot \textstyle\frac{1}{N + 3n} \geq \\
    & \nicefrac{1}{3} + n - \textstyle\frac{3\alpha n + 1}{N} \geq n +
    \nicefrac{1}{3} - \nicefrac{1}{9} > n + \nicefrac{1}{9} \text{.}
  \end{align*}
  The candidates from $D$ have even higher scores. Consequently, at
  most $n-1$ candidates have scores lower than $p$ and so $p$ cannot
  be a member of a winning committee.

  Thus, if there were a polynomial-time
  $\alpha$-approximation algorithm for our problem, then we could use
  it to decide the $\np$-hard \textsc{Restricted X3C} problem.
\end{proof}

\subsection{Mostly Hard Cases: GAV and RAV}

Unfortunately, for GAV and RAV we obtain an almost uniform set of
$\np$-hardness and inapproximability results. The only exception
regards (priced) adding approvals for the preferred candidate. 

\begin{theorem}
  Let $\calR$ be one of \textsc{GAV} and \textsc{RAV}, and let
  \textsc{Op} be one of \textsc{(\$)AddApprovals},
  \textsc{(\$)DeleteApprovals}, and \textsc{(\$)SwapApprovals}.  For
  each $\alpha > 1$, if $\p \neq \np$ then there is no polynomial-time
  $\alpha$-approximation algorithm for \textsc{$\calR$-Op-Bribery}.
  This also holds for \textsc{(\$)SwapApprovals} when we can move
  approvals only to the preferred candidate.
\end{theorem}

The somewhat involved proof of this theorem is inspired by a related
result of Bredereck et
al~\cite{bre-fal-nie-tal:c:multiwinner-shift-bribery}, for the case of
ordinal elections.

Nonetheless, the case of adding approvals for the preferred candidate
only is easy for both GAV and PAV (although for PAV in the priced
variant we only obtain a PTAS, i.e., a polynomial-time approximation
scheme).

\begin{theorem}
  When restricted to adding approvals to the preferred candidate
  only, \textsc{\{GAV,RAV\}-AddApprovals-Bribery} is in $\p$. For
  \textsc{GAV}, the priced variant of this problem is also in $\p$,
  whereas for \textsc{RAV} there is a PTAS for it.
\end{theorem}
\begin{proof}[sketch]
  Let $(C,V)$ be an election, let $p$ be the approved candidate, let
  $k$ be the committee size, and let $b$ be the budget.  We consider
  GAV first.  Since it proceeds in $k$ iterations, to ensure that $p$
  is selected, we first guess the iteration $\ell$ in which we plan
  for $p$ to be added to the committee. We execute GAV until the
  $\ell$'th round. Then we execute the following operation until
  either $p$ is to be selected in the $\ell$'th
  round\footnote{Technically, it is possible that by our actions $p$
    would be selected in an earlier round, but it does not affect the
    correctness of the algorithm.} or we exceed the budget (in which
  case, we try a different guess for $\ell$, or reject, if we ran out
  of possible guesses): We find a voter who does not approve any
  candidate in the so-far-selected committee and for whom the price
  for adding approval for $p$ is lowest; we add approval for $p$ for
  this voter.  Simple analysis confirms the running time and
  correctness of the algorithm.

  The algorithm for RAV is very similar: We also guess a round number
  where we plan for $p$ to be selected, and after simulating the
  algorithm until this round, we add the cheapest set of approvals
  guaranteeing that $p$ would be selected in this (or earlier)
  round. The only difference is that for the priced variant, this
  involves solving an instance of the \textsc{Knapsack} problem (each
  voter has a price for adding approval for $p$ and the number of
  points that we obtain by this approval, which is of the form
  $\nicefrac{1}{t}$, for some $t \in \{1,\ldots, \ell\}$). We can use
  a classic \textsc{Knapsack} PTAS for this task.
\end{proof}

Whether \textsc{RAV-\$AddApprovals-Bribery} is $\np$-hard when we can
add approvals for the preferred candidate only remains open (however,
we suspect that it does).


\subsection{FPT Algorithms}

While for several important special cases we either obtained direct
polynomial-time algorithms or polynomial-time approximation
algorithms, most of our problems are $\np$-hard and hard to
approximate within any constant factor. Fortunately, if either the
number of candidates or the number of voters is considered as the
parameter (i.e., can be assumed to be small), we have many $\fpt$
algorithms.


Indeed, for the unpriced setting and the parameterization by the number
of candidates all our problems are in $\fpt$. This follows by the
classic approach of formulating problems as integer linear
programs (ILPs) and applying Lenstra's
algorithm~\cite{len:j:integer-fixed}. Using the approach of Bredereck
et al.~\cite{BFNST15}
that combines Lenstra's algorithm with mixed
integer linear programming,
we also obtain $\fpt$ algorithms for the
priced cases of adding and deleting approvals.

\begin{theorem}
  For each $\calR$ in $\{$\textsc{AV}, \textsc{SAV}, \textsc{GAV},
  \textsc{RAV}, \textsc{CCAV}, \textsc{PAV}$\}$,
  \textsc{$\calR$-(\$)AddApprovals-Bribery} (also when we only add
  approvals for the preferred candidate),
  \textsc{$\calR$-(\$)DeleteApprovals-Bribery}, and
  \textsc{$\calR$-SwapApprovals-Bribery}
  are in $\fpt$ when parameterized by the number of candidates.
\end{theorem}

The reason why we do not obtain the result for
\textsc{\$Swap\-Approvals} is that the technique of Bredereck et
al.~\cite{BFNST15} requires that for each set of candidates $A$, and
for each possible set of bribery actions that can be applied to votes
approving exactly $A$---denote such votes as $V_A$---we have to be
able to precompute the cheapest cost of applying these actions to
exactly one vote from $V_A$, to exactly two votes from $V_A$,
etc. This is easy to do for (priced) adding and deleting approvals
because bribery actions are independent from each other.  Yet, this is
impossible for priced approval swaps as the lowest cost of moving an
approval from some candidate $c$ to some candidate $d$, within a vote
from $V_A$ may depend on what other swaps were performed before on
votes from $V_A$.
Nonetheless, we can handle \textsc{AV-\$SwapApprovals-Bribery}: In
this case it suffices to guess the winning committee and score $T$ of
its lowest-scoring member; then computing a bribery that ensure that
each member of the committee has score at least $T$ and each
non-member has score at most $T$ is easy though a min-cost/max-flow
argument.

\begin{proposition}
  \textsc{AV-\$SwapApprovals-Bribery} is in $\fpt$ when parameterized
  by the number of candidates.
\end{proposition}

For the parameterization by the number of voters, we use a more varied
set of approaches.  For the case where we add approvals for the
preferred candidate only, a simple exhaustive search algorithm is
sufficient, even for arbitrary prices.  Specifically, it suffices to
guess for which voters we add an approval for $p$, check that it is
within the budget, and that $p$ is then selected for some winning
committee. Recall that for the parameterization by the number of voters, winner
determination is in $\fpt$ for all our rules; for PAV and CCAV this
follows from the proof of Theorem~15 of Faliszewski et
al.~\cite{fal-sko-sli-tal:c:top-k-counting}. To simplify notation, we
will say that a rule has $\fpt(n)$ winner determination if there is an
$\fpt$ algorithm (parameterized by the number of voters) that checks
if a given candidate belongs to some winning committee.

\begin{theorem}
  For each rule $\calR$ with $\fpt(n)$ winner determination,
  \textsc{$\calR$-(\$)AddApprovals-Bribery} for the case where we add
  approvals for the preferred candidate only is in $\fpt$ when
  parameterized by the number of voters.
\end{theorem}

There are also general algorithms for the case of unpriced adding or
swapping approvals (not necessarily for $p$).  A \emph{unanimous}
voting rule is a voting rule for which if there is a candidate which
is approved by all the voters, then this candidate is in some winning
committee. (Note that all our rules are unanimous.)  A rule is
\emph{symmetric} if it treats all candidates and voters in a uniform
way (i.e., the results do not change if we permute the collection of
voters, and if we permute the set of candidates, then this analogously
permutes the candidates in the winning committees).  We say that two
candidates are of the same \emph{type} if they are approved by the
same voters; there are at most $2^n$ candidate types in an election
with $n$ voters (this idea of candidate types was previously used by
Chen et al.~\cite{che-fal-nie-tal:c:few-voters}).

\begin{theorem}
  For each symmetric, unanimous rule $\calR$ with $\fpt(n)$ winner
  determination, \textsc{$\calR$-AddApprovals-Bribery} and
  \textsc{$\calR$-SwapApprovals-Bribery} are in $\fpt$ when
  parameterized by the number of voters.
\end{theorem}

\begin{proof}
  Consider an instance of our problem with $n$ voters, where $p$ is
  the preferred candidate.  If the budget is at least $n$, then we
  accept because we can ensure that every voter approves $p$, and $p$
  is selected for some winning committee by unanimity.  So we assume
  that the budget is less than $n$.  We (arbitrarily) select $n$
  candidates of each candidate type present in the election (or all
  candidates of a given type, if there are fewer than $n$ of them).
  These at most $n \cdot 2^n$ candidates are the only ones which we
  allow to add (for \textsc{$\calR$-AddApprovals-Bribery}) or to swap
  between (for \textsc{$\calR$-SwapApprovals-Bribery}).  We can now
  check all possibilities of adding or swapping these candidates, and
  we accept if at least one leads to $p$ belonging to a winning
  committee and is within the budget.
\end{proof}

A similar technique works for \textsc{\$SwapApprovals}, for the case
where we are allowed to move approvals to the preferred candidate
only.  However, this time we cannot arbitrarily choose $n$ candidates
of each type, because we might choose candidates for whom moving the
approvals is too expensive.

\begin{theorem}
  For each symmetric, unanimous rule $\calR$ with $\fpt(n)$ winner
  determination, \textsc{$\calR$-\$SwapApprovals-Bribery} (for the
  case where we are allowed to move approvals to the preferred
  candidate only) is in $\fpt$ when parameterized by the number of
  voters.
\end{theorem}

\begin{proof}
  Consider an instance of our problem with $n$ voters, where $p$ is
  the preferred candidate.  Since we can move approvals to $p$ only,
  it follows that we cannot operate twice on the same voter and the
  number of operations in every solution is at most $n$.  Further, in
  a solution we might change at most $n$ candidates from their
  original types to some other types.  Consider two types, $\sigma$
  and $\sigma'$, and note that, as a corollary to the above
  observation, we have that at most $n$ candidates of type $\sigma$
  might change to type $\sigma'$ in the solution.  Therefore, for each
  pair of types $\sigma$ and $\sigma'$, we select at most $n$
  candidates of type $\sigma$ which are the cheapest to change to type
  $\sigma'$. These at most $n \cdot 2^{2n}$ candidates are the only
  ones which we allow to operate on.  We can now check all possible
  sets of move-approval-to-$p$ bribery actions on these candidates,
  and we accept if at least one leads to $p$ belonging to a winning
  committee and is within the budget.
\end{proof}

For CCAV and GAV, we use
the fact that they 
operate
on candidate types.
Then we provide a general $\xp$ result.

\begin{theorem}
  If $\calR$ is \textsc{CCAV} or \textsc{GAV}, then
  \textsc{$\calR$-(\$)Add}\-\textsc{Approvals-Bribery} and
  \textsc{$\calR$-(\$)DeleteApprovals-Bri\-bery} are in $\fpt$ 
  (parameterized by the number of voters).
\end{theorem}

\begin{proof}
  Consider some election with $n$ voters, committee size $k$, and
  where $p$ is the preferred candidate.  The crucial observation is
  that both for CCAV and GAV, the set of winning committees is fully
  determined by the candidate types present in the election
  (irrespective of the number of candidates of each type). This holds
  because whenever a candidate of some type is included in the
  committee, then adding another candidate of the same type will not
  change the committee's score.  In consequence, we can think of a
  winning committee as of a set of (at most $k$) candidate types.
  Candidate $p$ belongs to some winning committee if and only if its
  type belongs to some winning committee.

  If $k \geq n$, then we accept because for each candidate type there
  is a winning committee that includes it (it is always possible to
  choose at most $n$ candidate types so that the maximum number of
  voters approve the committee, and then we can add further, arbitrary
  candidate types).
  
  For the case where $k < n$, we proceed as follows. First, we guess
  candidate types that should be present in the election after the
  bribery (we also guess the type that $p$ shall have). Second, we
  compute the lowest cost of obtaining an election where exactly these
  candidate types are present (see below for an algorithm). Finally,
  we check if this cost does not exceed the budget and if there is a
  winning committee that includes $p$'s type. If so, we accept, and
  otherwise we try different guesses (and reject if no guess leads to
  acceptance).

  To compute the cost of transforming an election to one with exactly
  the guessed candidate types, we solve the following instance of the
  min-cost/max-flow problem~\cite{ahu-mag-orl:b:flows}.  For each
  candidate $c_i$ from the election, we create a node $c_i$.  We have
  a source node $s$ and we connect $s$ to each $c_i$ with an arc of
  capacity $1$ and cost $0$.  For each type $\sigma_j$ that we guessed
  to appear in the post-bribery election, we create a node $\sigma_j$.
  We have a target node $t$ and we connect each $\sigma_j$ to $t$ with
  an arc of cost $0$, infinite capacity, and the requirement that at
  least $1$ unit of flow passes through this arc (this ensures that
  each of the guessed types actually appears in the election after the
  bribery). We connect each $c_i$ to each $\sigma_j$ by an arc with
  capacity $1$ and cost equal to the price of changing the type of
  $c_i$ to the type $\sigma_j$. (For \textsc{\$AddApprovals} and
  \textsc{\$DeleteApprovals} we can indeed compute these costs
  independently for each candidate and each candidate type, using
  infinite costs to model impossible transformations).  This network
  contains at most $O(m+2^n)$ nodes, where $m$ is the number of
  candidates and $n$ is the number of voters. Since there is a
  polynomial-time algorithm that solves the min-cost/max-flow
  problem~\cite{ahu-mag-orl:b:flows} (i.e., that finds a minimum-cost
  flow that satisfies all the arc requirements and moves a given
  number of units of flow, $m$ in our case, from the source to the
  sink), this algorithm computes the desired cost in $\fpt$ time with
  respect to the number of voters.
\end{proof}



\begin{theorem}
  For each rule and bribery problem studied in this paper, the problem
  is in $\xp$ both for the parameterization by the number of
  candidates and voters.
\end{theorem}


\section{Related Work}\label{sec:related}
The bribery family of problems was introduced by Faliszewski et
al.~\cite{fal-hem-hem:j:bribery}, but in that work the authors mostly
(but not only) 
focused on the case where after ``buying'' a vote it is
possible to change it arbitrarily. Bribery problems where each local
change in a vote is accounted for separately, were first studied by
Faliszewski et al.~\cite{fal-hem-hem-rot:j:llull} (for irrational
votes) and then by Faliszewski~\cite{fal:c:nonuniform-bribery} (in
particular, for the single-winner approval setting, by allowing
different costs of moving approvals between candidates) and by Elkind
et al.~\cite{elk-fal-sli:c:swap-bribery,elk-fal:c:shift-bribery} (for
the standard ordinal model, in the \textsc{Swap Bribery} problem by
assigning different costs for swapping adjacent candidates in a
preference order, and in the \textsc{Shift-Bribery} problem by
assigning different costs for shifting the preferred candidate forward
in preference orders). \textsc{Swap Bribery} and \textsc{Shift
  Bribery} were then studied by a number of authors, including Dorn
and Schlotter~\cite{dor-sch:j:parameterized-swap-bribery} and
Bredereck et
al.~\cite{bre-che-fal-nic-nie:j:parametrized-shift-bribery}.
Bredereck et al.~\cite{bre-fal-nie-tal:c:multiwinner-shift-bribery}
studied the complexity of \textsc{Shift Bribery} for multiwinner
elections (their paper is very close to ours).

Our work was inspired by that of Aziz et
al.~\cite{azi-gas-gud-mac-mat-wal:c:multiwinner-approval} on the
complexity of winner determination and strategic voting in
approval-based elections. In addition, Aziz et
al.~\cite{azi-bri-con-elk-fre-wal:j:justified-representation}
introduced the notion of justified representation and argued why rules
such as PAV, CCAV, and GAV should be very effective
for 
achieving proportional representation or, at least, diversity within
the committee.  For more details regarding AV, SAV, PAV, and RAV, 
we point the reader
to the work of Kilgour~\cite{kil-handbook}. PAV, RAV, CCAV, and GAV
were introduced in the 19th century by Thiele~\cite{Thie95a}
(GAV-style rules attracted attention after Lu and
Boutilier~\cite{bou-lu:c:chamberlin-courant} considered them in the
ordinal setting).  CCAV 
is a variant of the
Chamberlin--Courant rule~\cite{cha-cou:j:cc}, 
but for the approval
setting; 
studied, e.g., by Procaccia et
al.~\cite{pro-ros-zoh:j:proportional-representation} and Betzler et
al.~\cite{bet-sli-uhl:j:mon-cc}.

Other closely related papers include that of Meir et
al.~\cite{mei-pro-ros-zoh:j:multiwinner} (on the complexity of
manipulation and control for multiwinner rules) and those of Magrino
et al.~\cite{mag-riv-she-wag:c:stv-bribery} and
Xia~\cite{xia:margin-of-victory} (on using bribery to quantify chances
of election fraud; 
we also use
bribery for post-election analysis).

\section{Outlook}\label{sec:outlook}
We believe that our most important contribution is conceptual: We
propose to use bribery problems to measure how well each candidate
performed in an election. While many (yet, not all) of our results are negative,
we show an extensive set of FPT 
results. 
It is important to verify how efficiently can our problems be solved
in practice. 

\section*{Acknowledgments} 
Piotr Faliszewski was supported by the National Science Centre, Poland,
under project 2016/21/B/ST6/01509.  Piotr Skowron was supported by the
ERC grant 639945 (ACCORD) and by a Humboldt Research Fellowship for
Postdoctoral Researchers.  Nimrod Talmon was supported by a
postdoctoral fellowship from I-CORE ALGO.

\bibliographystyle{abbrv}
\bibliography{bib}

\end{document}